\documentclass[conference]{IEEEtran}
\IEEEoverridecommandlockouts
\usepackage{cite}
\usepackage{amsmath,amssymb,amsfonts}
\usepackage{graphicx}
\graphicspath{{figures/}}
\usepackage{textcomp}
\usepackage{xcolor}
\usepackage[a4paper, total={184mm,239mm}]{geometry}
\def\BibTeX{{\rm B\kern-.05em{\sc i\kern-.025em b}\kern-.08em
    T\kern-.1667em\lower.7ex\hbox{E}\kern-.125emX}}
    
\usepackage{CJKutf8}
\usepackage{changes}

\usepackage{gensymb}
\usepackage{bm}

\usepackage{amsmath,amssymb}
\usepackage{amsthm}

\usepackage{balance}

\usepackage{graphicx}
\graphicspath{{figures/}}
\usepackage{subcaption}
\usepackage[export]{adjustbox}
\usepackage{flushend}

\usepackage{listings}
\usepackage{color}
\usepackage{soul}
\usepackage{tikz}
\newcommand*{\circled}[1]{\lower.7ex\hbox{\tikz\draw (0pt, 0pt)%
    circle (.5em) node {\makebox[1em][c]{\small #1}};}}

\usepackage{algorithm}
\usepackage{algpseudocode}
\newcommand\cmd[1]{\texttt{\small#1}}

\theoremstyle{plain}
\newtheorem{theorem}{Theorem}

\newtheorem{lemma}{Lemma}

\newtheorem*{proof}{Proof}

\theoremstyle{definition}
\newtheorem{definition}{Definition}
 
\theoremstyle{remark}

\definecolor{mygreen}{rgb}{0,0.6,0}
\definecolor{mygray}{rgb}{0.5,0.5,0.5}
\definecolor{mymauve}{rgb}{0.58,0,0.82}

\usepackage{algpseudocode}
\usepackage{algorithm}

\usepackage{url}
\usepackage{booktabs} % For formal tables
\usepackage{multirow}
\usepackage{array}
\newcolumntype{I}{!{\vrule width 1.2pt}}
\newlength\savedwidth

\newlength\savewidth
\newcommand\shline{\noalign{\global\savewidth\arrayrulewidth
                           \global\arrayrulewidth 1.2pt}%
                  \hline
                  \noalign{\global\arrayrulewidth\savewidth}}
\usepackage{makecell}

\usepackage{geometry}
\begin{document}
\begin{CJK*}{UTF8}{gbsn}

\title{%An Efficient NPN Classifier Utilizing \\Cofactor and Sensitive Signatures \\
%An Efficient Canonical Vector for NPN Classification Utilizing Cofactor and Sensitive Signatures\\
% A New NPN Classifier Utilizing Face and Point Characteristics of Boolean Functions
Rethinking NPN Classification from Face and Point Characteristics of Boolean Functions
\thanks{\IEEEauthorrefmark{1}These authors contributed equally to this work.}
}
\author{\IEEEauthorblockN{
Jiaxi~Zhang$^{1}$\IEEEauthorrefmark{1},
Shenggen~Zheng$^{2}$\IEEEauthorrefmark{1}, 
Liwei~Ni$^{2,3}$\IEEEauthorrefmark{1},
Huawei~Li$^{3,4}$
and~Guojie~Luo$^{1}$
}
\IEEEauthorblockA{
$^1$Center for Energy-Efficient Computing and Applications, Peking University, Beijing, China}
\IEEEauthorblockA{
$^2$Peng Cheng Laboratory, Shenzhen, China
}
\IEEEauthorblockA{
$^3$University of Chinese Academy of Sciences, Beijing, China
}
\IEEEauthorblockA{
$^4$Institute of Computing Technology, Chinese Academy of Sciences, Beijing, China
}

\IEEEauthorblockA{
Email: zhangjiaxi@pku.edu.cn, 
zhengshg@pcl.ac.cn,
nlwmode@gmail.com,
lihuawei@ict.ac.cn,
gluo@pku.edu.cn
}}
% \and
% \and
% \and
% \and

\maketitle

\begin{abstract}

NPN classification is an essential problem in the design and verification of digital circuits.
Most existing works explored variable symmetries and cofactor signatures to develop their classification methods.
However, cofactor signatures only consider the face characteristics of Boolean functions.
In this paper, we propose a new NPN classifier using both face and point characteristics of Boolean functions, including cofactor, influence, and sensitivity.
The new method brings a new perspective to the classification of Boolean functions.
The classifier only needs to compute some signatures, and the equality of corresponding signatures is a prerequisite for NPN equivalence.
Therefore, these signatures can be directly used for NPN classification, thus avoiding the exhaustive transformation enumeration.
The experiments show that the proposed NPN classifier gains better NPN classification accuracy with comparable speed.
%Moreover, for fundamental functions in open-source circuit benchmarks, the semi-canonical forms gain an exact classification when the input size $n \leq 7$.

\end{abstract}

%\IEEEpeerreviewmaketitle

\begin{IEEEkeywords}
NPN classifier, cofactor, influence, sensitivity
\end{IEEEkeywords}

\section{Introduction}

Classification of Boolean functions groups a set of Boolean functions into equivalent classes.
Negation-Permutation-Negation~(NPN) equivalence is the most frequently used one regarding the transformations of input negation, input permutation, and output negation.
It has significant applications in logic synthesis, technology mapping, and verification.
Boolean functions that are NPN equivalent define an NPN equivalence class.

Boolean matching and classification have been widely studied in the past decades.
Previous methods could be classified into three categories: algorithms based on Boolean satisfiability~(SAT), algorithms utilizing search with signature pruning, and algorithms based on canonical forms.
SAT-based methods can handle Boolean functions with a large number of input variables, but they are slow because of the NP-completeness~\cite{soeken2016heuristic}.
Algorithms utilizing search with signature pruning are usually used for pair-wise matching.
A signature of a Boolean function is a compact representation that characterizes some intrinsic structures and serves as a necessary condition for Boolean matching.
Signatures derived from row sums~\cite{chai2006building} and cofactor~\cite{chai2006building,abdollahi2008symmetry,agosta2009transform,zhou2019fast} are two common types that have been explored extensively.
Zhang {\it et al.}~\cite{zhang2021enhanced} explore sensitivity signatures for further pruning.
Spectra like Walsh~\cite{clarke1993spectral} and Haar~\cite{thornton2002logic} have also been used as signatures for Boolean matching. 

Algorithms based on canonical forms are best manifested in NPN classification.
A canonical form is a representative of NPN equivalent Boolean functions, and two functions match if and only if their canonical forms are identical.
They work by designing a complete and unique canonical form of the Boolean functions and then try computing the canonical form for each Boolean function to check for NPN equivalence.
Many works construct canonical form considering phase assignment~\cite{tsai1997boolean}, variable symmetries~\cite{abdollahi2005new,zhang2019efficient} and high-order symmetries~\cite{kravets2000generalized,huang2013fast,zhou2019fast}.
Abdollahi {\it et al.}~\cite{abdollahi2008symmetry} utilize cofactor signatures to define signature-based canonical forms.
Zhou {\it et al.}~\cite{zhou2020fast} combine cofactor signatures and different types of symmetries to design hybrid canonical forms.

From the hypercube view of the Boolean functions, the cofactors only include the face characteristics of the hypercube. 
When designing an NPN classifier, it is challenging to ensure classification accuracy if only the face characteristics are considered.
Exhaustive transformation enumerations are always required to improve classification accuracy or achieve exact classification.
In this paper, we develop a new NPN classification method that considers both face characteristics and point characteristics of Boolean functions, which are sensitivity~\cite{cook1982bounds} and influence~\cite{kahn1988influence}.  
Intuitively, a cofactor considers a face of the hypercube and counts how many points take the same value in the face,  while sensitivity considers a point of the hypercube and counts up how many adjacent points take a different value with the point.   
The main contributions are summarized as the following: 
\begin{itemize}
    \item We introduce Boolean sensitivity and influence into NPN classification. 
    We deeply analyze the relationship between these two characteristics and the cofactor to illustrate their different properties.
    They bring a new perspective to the NPN classification.
    \item We design some signature vectors based on these two characteristics and give some proofs for these vectors to guide NPN equivalence checking.
    \item We develop a new NPN classifier based on these signature vectors.
    After signatures computation, the classifier can directly get NPN classes without transformation enumeration.
    Meanwhile, the classifier has stable runtime; it does not suffer from variance for different Boolean function sets.
\end{itemize}

%The rest of the paper is organized as follows. 
%Section~\ref{sec: background} introduces the basic concepts and two minterm statistics signatures.
%Section~\ref{sec: variability} summarizes three key concepts in Boolean algebra, cofactor, sensitivity, and influence, and how to construct the semi-canonical form.
%Section~\ref{sec:prove} gives the matching procedure.
%Section~\ref{sec:canonical} .
%Section~\ref{sec: evaluation} shows the experimental results and ends with the conclusion in Section~\ref{sec: conclusion}. 

\section{Face and Point Characteristics}
\label{sec:background}
This section shows some commonly used notations and concepts and then introduces the three characteristics used in our NPN classifier.

\subsection{Notations and Basic Concepts}%\added[comment={最好用表格}]{}
\label{sec:basicconcept}

An $n$-variable Boolean function, $f(X): \{0,1\}^n \rightarrow \{0,1\}$, maps a binary word $X=(x_1,x_2,\cdots,x_n)$ of width $n$ into a single binary value.
A variable $x_i$ or its complement $\bar{x_i}$ in $f$ is called a literal, and $i$ denoted the index.
A \emph{minterm} is a conjunction of $n$ literals of different variables.

Truth table $T(f)$, a binary string of $2^n$ bits, is a commonly-used representation of Boolean function $f$.
The $i$-th bit of $T(f)$ is equal to $f((i)_2)$, where $(i)_2$ is the little-endian binary code of integer $i$.
%Thus, $T(f)=(f((2^n-1)_2),...,f((1)_2),f((0)_2))$.
A subgraph of a hypercube can also represent a Boolean function.
The hypercube $Q_n$ is a graph of order $2^n$, whose vertices include all minterms and whose edges connect vertices that differ in exactly one variable.
Boolean function $f$ can be represented as the induced subgraph of $Q_n$.
Fig.~\ref{fig:func1} shows the hypercube $Q_3$, and the $\bullet$ nodes, as well as edges between these nodes, construct the induced subgraph of a 3-Majority logic.

%In this paper, we denote 
An NP transformation of a Boolean function is composed of variables negations and permutations.
Negation, denoted as $\neg$, replaces a variable by its complement~(e.g., $x_1 \to \neg x_1$).
Besides, we denote $(\neg)$ as a selective negation.
For simplicity, we denote $(\neg)X=(\neg)x_1(\neg)x_2\cdots(\neg)x_n$ to describe the selective negation of the word $X$ (e.g., for $(\neg)(x_1x_2) = x_1\overline{x_2}$, we have $(\neg)x_1 = x_1$ and $(\neg)x_2 = \overline{x_2}$).
Permutation, denoted as $\pi$, is a reorder of variables~(e.g., $\pi(x_1x_2) = x_2x_1$).
%Thus, two Boolean functions $f$ and $g$ are NPN-equivalent, denoted as $f \cong g$,\added[comment={后文有用到吗}]{} if there exists an NP transformation $\Gamma$ such that $\Gamma (f)=(\neg)g(x)$.\added[comment={?}]{}
Besides, we denote $X_{(i)}$ as the $i$-th minterm, and $X^i$ as negating the $i$-th variable in $X$.

The \emph{satisfy count} of a function $f$ is the number of nodes in the induced subgraph, denoted as $|f|$.
An $n$-input $f$ is called \emph{balanced} if $|f|=|\bar{f}|=2^{n-1}$.
Fig.~\ref{fig:npnequ} shows three balanced 3-inputs Boolean functions.

\begin{figure}[tbp]
\centering
\begin{subfigure}[b]{0.32\columnwidth}
\centering
\includegraphics[scale=0.30]{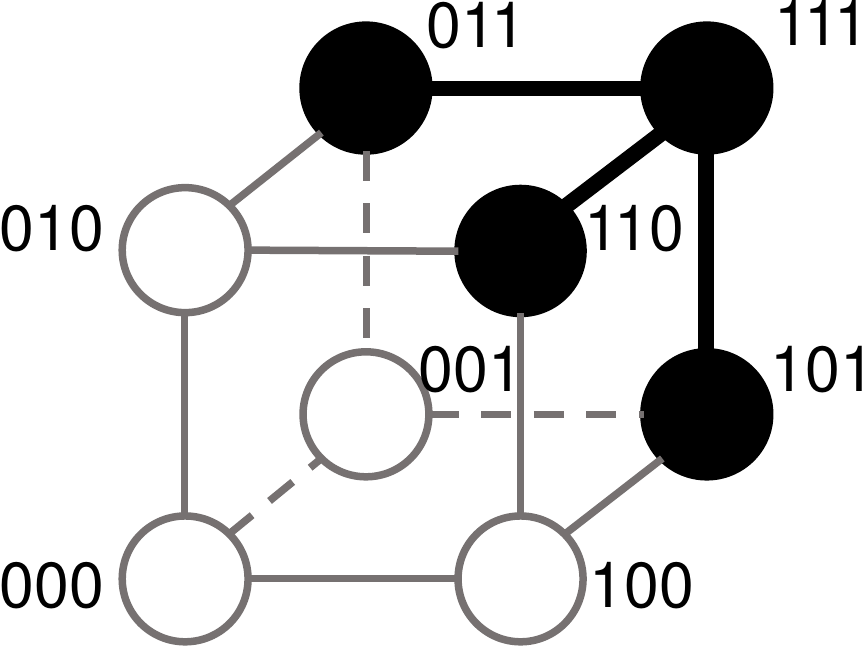}
%\caption{3-Majority Maj($x_1$, $x_2$ ,$x_3$)= \\ $x_1 x_2 + x_2 x_3 + x_1 x_3$)}
\caption{3-Majority $f_1$.}
\label{fig:func1}
\end{subfigure}
\begin{subfigure}[b]{0.32\columnwidth}
\centering
\includegraphics[scale=0.30]{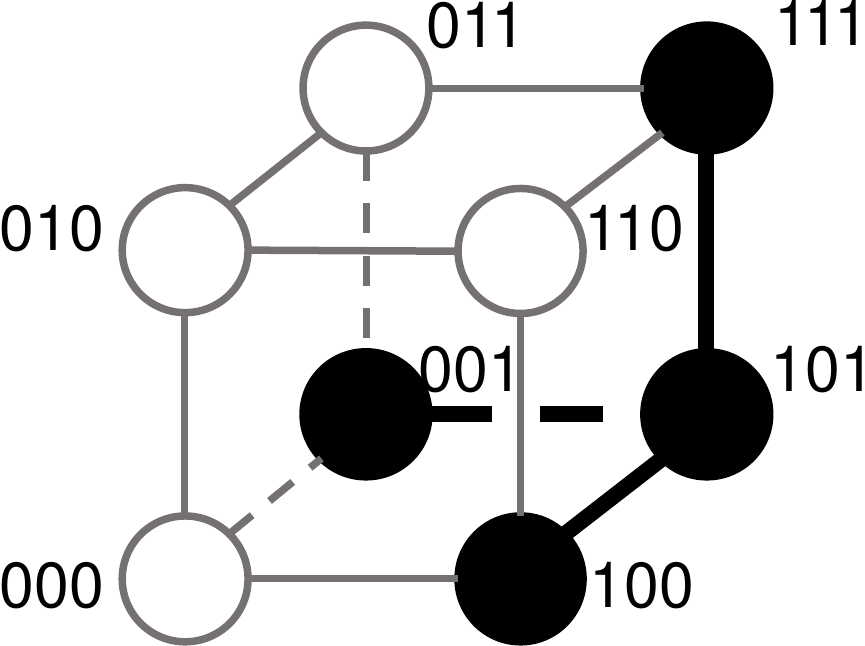}
\caption{$f_2$.}
\label{fig:func2}
\end{subfigure}
\begin{subfigure}[b]{0.32\columnwidth}
\centering
\includegraphics[scale=0.30]{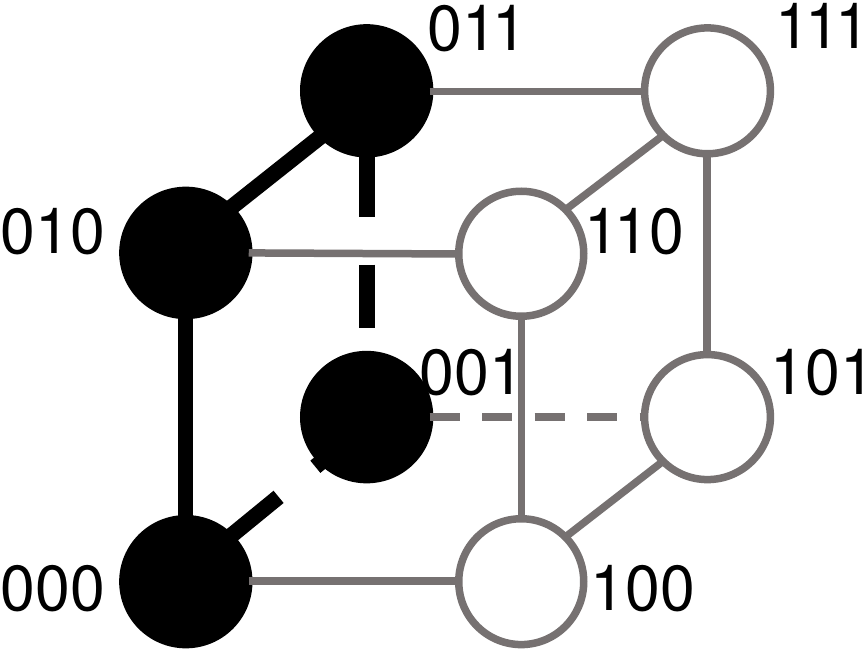}
\caption{$f_3$.}
\label{fig:func3}
\end{subfigure}
\caption{Hypercubes of three 3-variable Boolean functions. 3-majority logic $f_1$ and $f_2$ are NPN equivalent, and their induced subgraphs(bolded) are isomorphic. $f_2$ and $f_3$ are not NPN equivalent, and their induced subgraphs are non-isomorphic.}
\label{fig:npnequ}
\vspace{-1em}
\end{figure}

%For an $n$-variable Boolean function, there are $2^n$ ways of inputs negation and $n!$ ways of inputs permutation.
%Besides, there are two polarities of the function derived by complementing its output.
%Thus, the complexity of the exhaustive enumeration method for Boolean matching is $O(2^{n+1}n!)$.

\subsection{Face Characteristic -- Cofactor}
\label{sec:cofactor}
The cofactor is derived from a Boolean function by substituting constant values for some input variables, and it has been well explored in Boolean matching and  classification~\cite{chai2006building,abdollahi2008symmetry,agosta2009transform,zhou2019fast,zhang2019efficient} in the past decades.

\begin{definition}
\emph{(cofactor)}.
The \emph{cofactor} of $f$ with respect to literal $x_i$ and $\overline{x_i}$ are denoted as $f_{x_i=1}$ and $f_{x_i=0}$, respectively.
\end{definition}

\begin{definition}
\emph{(cofactor signatures)}.
The \emph{cofactor signatures} are the satisfy count and the satisfy counts of the cofactors.
Particularly, the satisfy count of a Boolean function is called the \emph{0-ary cofactor signature}.
The \emph{1-ary cofactor signatures} are a set of satisfy counts of the cofactors with respect to each literal.
The \emph{higher-ary cofactor signatures}~(a.k.a higher-order cofactor signatures) are composed of the satisfy counts of the cofactors with respect to a set of variables.
\end{definition}

A face in the hypercube represents a cofactor with respect to one variable or multiple variables of a Boolean function.
In Fig.~\ref{fig:cofvisual}, the blue face is $f_{x_1=0}$.
Moreover, the blue face in Fig.~\ref{fig:cof2visual} represents the cofactor $f_{x_1x_2=00}$.
Therefore, cofactor signatures are the number of 1-minterms on the face in the hypercube.
They contain the \textbf{face characteristics} of Boolean functions.
Higher-ary cofactor signatures correspond to lower-ary faces.
Symmetry can be deduced by cofactor, which can also be seen as a property of faces in a Boolean function.

\begin{figure*}[htbp]
\setlength{\belowcaptionskip}{0pt}
\centering
\begin{subfigure}[b]{0.24\textwidth}
\centering
\includegraphics[scale=0.35]{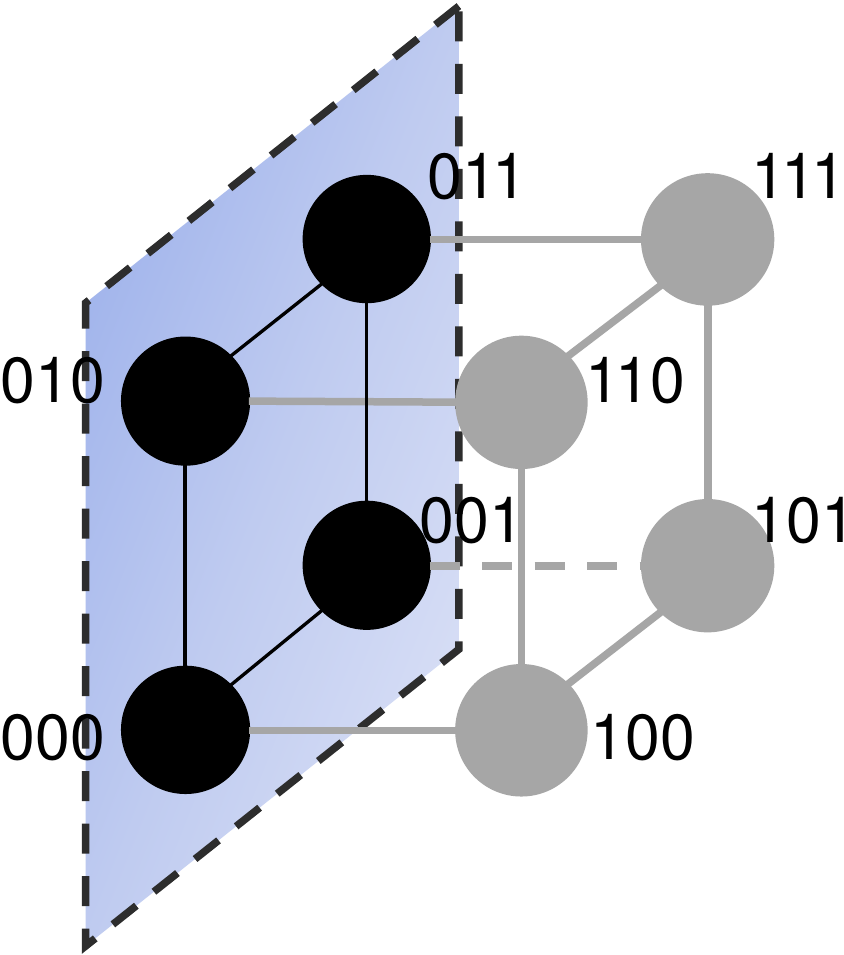}
\caption{cofactor.}
\label{fig:cofvisual}
\end{subfigure}
\begin{subfigure}[b]{0.24\textwidth}
\centering
\vspace{-0.5em}
\includegraphics[scale=0.35]{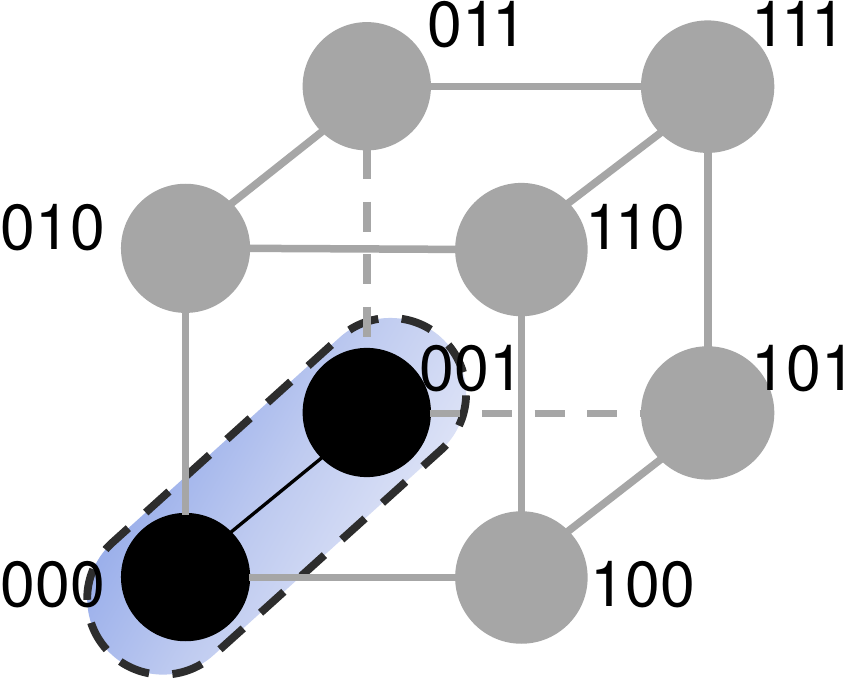}
\caption{2-ary cofactor.}
\label{fig:cof2visual}
\end{subfigure}
\begin{subfigure}[b]{0.24\textwidth}
\centering
\includegraphics[scale=0.35]{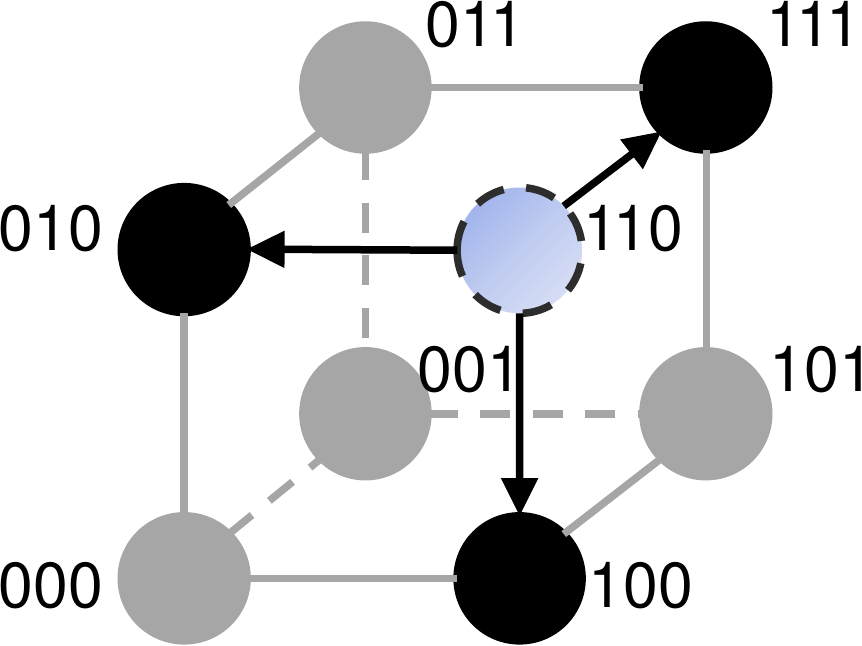}
\caption{sensitivity.}
\label{fig:senvisual}
\end{subfigure}
\begin{subfigure}[b]{0.24\textwidth}
\centering
\includegraphics[scale=0.35]{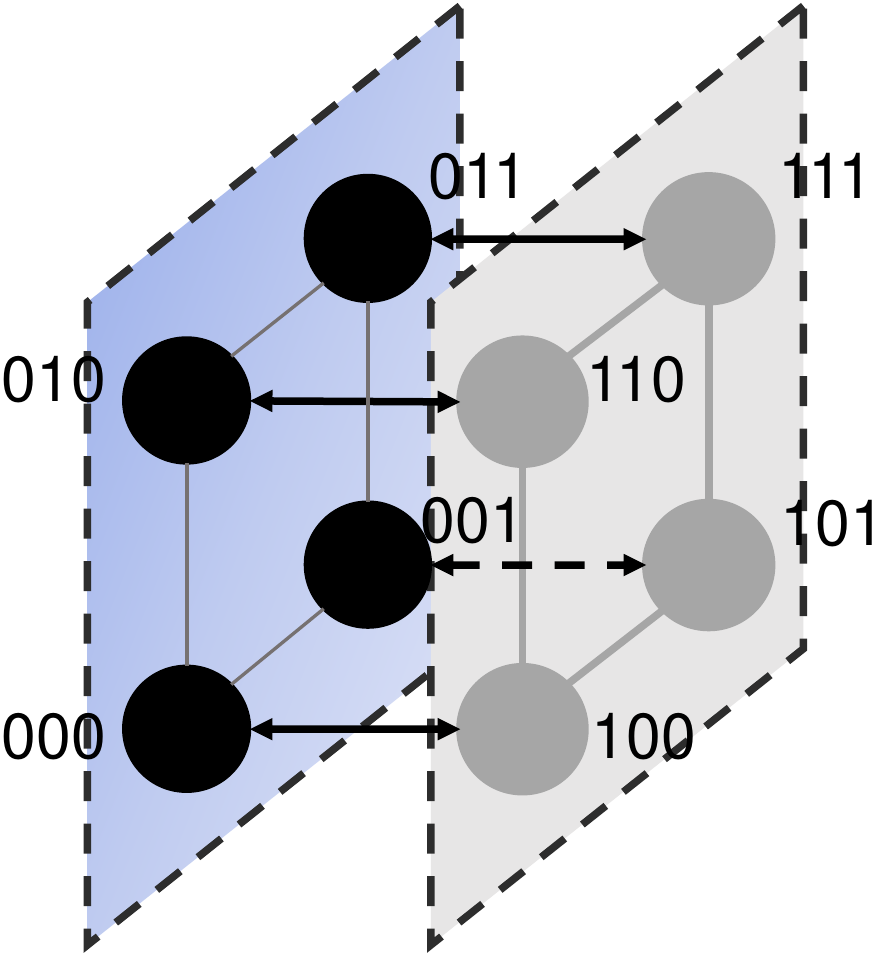}
%\caption{$sen(f)=3, s^1(f)=3$\\ \scriptsize{$OSV(f)=\{3,3,1,1,\linebreak[0]1,1,\linebreak[0]1,1\}$} \\ $OSV^1(f)=\{3,3\}$}
\caption{influence.}
\label{fig:infvisual}
\end{subfigure}
%\subfigbottomskip{10pt}
\caption{Cofactor, Influence and Sensitivity Signatures Visualization.}
\label{fig:visual}
\vspace{-1em}
\end{figure*}

\subsection{Point Characteristic -- Sensitivity}

%Sensitivity was introduced as a simple combinatorial complexity measure for Boolean functions. 
Sensitivity considers a hypercube point and counts how many adjacent points take a different value from this point.

\begin{definition}\label{def:sensitive}
(\emph{sensitive}).
Given a word $X$, Boolean function $f$ is \emph{sensitive} at literal $x_i$ for the word $X$, if the output flips when the $x_i$ flips~(i.e., $f(X) \neq f(X^i)$).
\end{definition}

Take the word $X$=\emph{100} for the 3-Majority logic $f_1$ in Fig.~\ref{fig:func1} as an example.
If the second bit flips,  $f_1$ will also flip.
Thus, $f_1$ is sensitive at $x_2$ for the word \emph{100}.

\begin{definition}\label{def:sensitivity}
(\emph{sensitivity}).
The \textit{sensitivity} of $f$ on the given word $X$, a.k.a.~\textit{local sensitivity}, is the number of input literals that are sensitive for $X$: $sen(f,X) = |{i:f(X) \neq f(X^{i})}|$.
Further we have the \emph{sensitivity} of $f$ as $sen(f) = \max\{sen(f,X): X \in \{0,1\}^n\}$.
The $\emph{0}$-\emph{sensitivity} of $f$ as $sen^0(f) = \max\{sen(f,X): X \in \{0,1\}^n, f(X) = 0\}$ and the $\emph{1}$-\emph{sensitivity} of $f$ as $sen^1(f) = \max\{sen(f,X): X \in \{0,1\}^n, f(X) = 1\}$.
\end{definition}

Sensitivity reflects the relation between neighboring points.
In Fig.~\ref{fig:senvisual}, the local sensitivity $sen(f, 110)$ indicates the property between the blue points and their neighboring points.
Therefore, sensitivity signatures contain the \textbf{point characteristics} of a Boolean function.

\subsection{Point-Face Characteristic -- Influence}

Influence is the probability that points in one face have a different value than the corresponding points in the opposite face.

\begin{definition}\label{def:inf}
(\emph{influence}).
The \emph{influence} of input $x_i$ on Boolean function $f$ is defined to be the probability that $f$ is sensitive at $x_i$ for a word $X$: $inf(f, i)=\mathop{Pr}\limits_{X \in \{0,1\}^n} [f(X) \neq f(X^i)] = \frac{1}{2^n} |f(X) \neq f(X^i): X\in \{0,1\}^n|$
\footnote{For convenience, we denote in the rest of this paper that  $inf(f, i)=\frac{1}{2}|f(X) \neq f(X^{i}): X\in \{0,1\}^n|$.
It is clear that $|f(X)\neq f(X^i) : X\in {0,1}^n|$ is an even integer. 
For example, if $f(000)\neq f(100)$, then $f(100)\neq f(000)$.   
Once the factor $\frac{1}{2^n}$ is removed, $inf(f,i)$ will be an integer. 
It is easier to compute integers than floating point numbers.}.
Furthermore, the \emph{total influence} of $f$ can be further defined as $inf(f)=\sum_{i=1}^n inf(f,i)$.
\end{definition}

% \begin{definition}
% (\emph{block influence}).
% The \textit{block influence}~(a.k.a higher-ary influence) of index set $S_i$ on Boolean function $f$ is defined to be the probability that $S_i$ is sensitive for a word $X$: $binf(f, S_i)=\mathop{Pr}\limits_{X \in \{0,1\}^n} [f(X) \neq f(X^{S_i})]$.
% The \textit{total block influence} of $f$ can be further defined as $binf(f)=\sum_{i=1}^n binf(f,S_i)$.
% If we restrict the block size to be at most $\ell$, then we obtain $\ell$-\emph{block influence} of the function $f$, denoted as $binf_{\ell}(f)$.
% \end{definition}

% \begin{example}
% Figure~\ref{fig:examples} shows $inf_0$ and $OIV_1$of five different 3-variable Boolean functions.
% For Figure~\ref{fig:sensitivity1}, $f(X=000) \neq f(X=000^{x_1})$, $f(X=001) \neq f(X=001^{x_1})$, $f(X=010) \neq f(X=010^{x_1})$, $f(X=011) \neq f(X=011^{x_1})$, $f(X=100) \neq f(X=100^{x_1})$, $f(X=101) \neq f(X=101^{x_1})$, $f(X=110) \neq f(X=110^{x_1})$, $f(X=111) \neq f(X=111^{x_1})$.
% Thus, $inf(f, x_1)=\frac{1}{2} \times 8 = 4$.
% Similarly, $inf(f, x_2)=0$ and $inf(f, x_3)=0$.
% Then, $inf(f)=4$, and $OIV_1(f)=\{0,0,4\}$.
% \end{example}

Boolean influence derives from the sensitive definition.
The sensitive property captures the relation of neighboring points.
Influence indicates the sensitive properties between two opposite faces.
It calculates the number of minterms in a face with different values compared to the opposite face. 
In Fig.~\ref{fig:infvisual}, the influence of variable $x_1$ reflects on the blue and the grey face.
Therefore, influence signatures contain the \textbf{point-face characteristics} of a Boolean function.

From the above analysis, it is evident that the structural information of a Boolean function contained in the influence and sensitivity signatures is of a different perspective compared to the cofactor signatures.
Most previous works only considered the role of cofactor signatures in constructing the NPN classification method while ignoring the influence and sensitivity features.
These two characteristics have great potential to guide NPN classification.
We will explore them in the following two sections.

\section{Signature Vectors and NPN Equivalence}
\label{sec:prove}
In this section, we first design several signature vectors from the point and face characteristics of Boolean functions and then give some theorems and their proofs of these signature vectors and NPN equivalence.

\subsection{Signature Vectors}
We can further define several signature vectors from the definition of the face and point characteristics.

\begin{definition}
\emph{(ordered cofactor vector)}.
The \emph{\ 1-ary ordered cofactor vector} of an $n$-variable Boolean function $f$ is $\mathit{OC}\, V_1=\{ \left|f_{z=v}\right|: z\in\{x_1,x_2,...,x_n\}, v\in\{0,1\}\}_\le$, where $\{\cdot\}_\le$ is the sorted multi-set (of all cofactors' satisfy counts) in non-decreasing order and $|\mathit{OC}\, V_1|=2n$.
%Furthermore, the \emph{$\ell$-ary ordered cofactor vector} of an $n$-variable Boolean function $f$ is the sorted multi-set $\mathit{OC}\,V_\ell=\{ \left|f_{z=v}\right|: z\in\{x_1,x_2,...,x_n\}_\ell$, $v\in\{0,1\}^\ell\}_\le$, and $|\mathit{co}\,f_\ell| = {n \choose \ell} \cdot 2^\ell$.\added[comment={$z$所属的集合没有解释}]{}
Furthermore, the \emph{$\ell$-ary ordered cofactor vector} of an $n$-variable Boolean function $f$ is the sorted multi-set $\mathit{OC}\,V_\ell=\{ \left|f_{z=v}\right|: z\in\{x_1,x_2,...,x_n\}_\ell$, $v\in\{0,1\}^\ell\}_\le$, where $Z_\ell=\{z \subseteq Z : |z|=\ell\}$ and $|\mathit{co}\,f_\ell| = {n \choose \ell} \cdot 2^\ell$.
\end{definition}

\begin{definition}
(\emph{ordered influence vector}).
The \emph{ordered influence vector} of Boolean function $f$ is $\mathit{OIV}(f)=\{ inf(f,z): z\in\{x_1,x_2,...,x_n\} \}_\le$.
%\deleted{, where $\{\cdot\}_\le$ is the sorted multi-set of satisfy counts in non-decreasing order}.}{For all indexes $\bm{x}$ in $f(x)$, we denote $OIV(f)=\Big(inf(f,x_{(0)}),...,inf(f,x_{(n-1)})\Big)$ such that $inf(f,x_{(0)}) \ge \cdots \ge inf(f,x_{(n-1)})$ as the \emph{ordered influence vector} of function $f$.}
%The \emph{$\ell$-ary ordered influence vector} ($\ell \geq 2$) of Boolean function $f$ is the sorted multi-set of block influence $\mathit{OIV}_\ell(f)=\{inf(f, z): z\in\{x_1,x_2,...,x_n\}_\ell \}_\le$.
%\deleted{, where $S_\ell=\{s \subseteq S : |s|=\ell\}$}.
\end{definition}

\begin{definition}\label{def:osv}
(\emph{ordered sensitivity vector}).
For all words $X$ in truth table $T(f)$, we denote the sorted multi-set $OSV(f)=\{sen(f,X): X \in\{0,1\}^n\}_\le$ as the \emph{ordered sensitivity vector} of function $f$.
%{$OSV(f)=\Big(sen(f,x^{(1)}),...,sen(f,x^{(N)})\Big)$ such that $sen(f,x^{(1)}) \le \cdots \le sen(f,x^{(N)})$ as the \emph{ordered sensitivity vector} of function $f$, where $N=|X|$ is the total number of words.}
Similarly, we can define $OSV^0(f)=\{sen(f,X): X \in\{0,1\}^n, f(X)=0\}_\le$ as the \emph{ordered $\mathit{0}$-sensitivity vector} and $OSV^1(f)=\{sen(f,X): X \in\{0,1\}^n, f(X)=1\}_\le$ as the \emph{ordered $\mathit{1}$-sensitivity vector}.
Obviously, we have $OSV(f)$ = $\{OSV^1(f), OSV^0(f)\}_\le$.
\end{definition}

\begin{definition} \label{d7}
(\emph{sensitivity distance}).
\emph{Hamming distance}
$h(X,Y)$ is a metric for comparing two binary strings $X$ and $Y$. 
It is the number of bit positions in which $X$ and $Y$ differ. 
The \emph{sensitivity distance} is defined as the Hamming distance of two words $X$ and $Y$ that have the same local sensitivity, denoted as a tuple $\{ \langle sen(f, X), sen(f, Y), h(X, Y)  \rangle | sen(f, X)$=$sen(f, Y) \}$.
\end{definition}

\begin{definition} \label{d8}
(\emph{ordered sensitivity distance vector}).
For a given $n$-variable Boolean function $f$, we define $OSDV(f)$ = $(\sigma_0, \sigma_1, \cdots, \sigma_n)$, where $\sigma_i$ = $( \delta_{i1}$, $\delta_{i2}$, $ \cdots$, $\delta_{in})$ and $\delta_{ij}$ = $|\{(X,Y) : sen(f,X)$=$sen(f,Y)$=$i$, $h(X,Y)$=$j$, and $ X < Y\}|$. $\delta_{ij}$ is the number of pairs $(X,Y)$ with sensitivity $i$ and distance $h(X,Y)$=$j$.  
Similarly, we can define $OSDV^1(f)$ and $OSDV^0(f)$ based on Definition~\ref{def:sensitivity}.
%Obviously, we have $OSDV(f)$ = $\{OSDV^1(f), OSDV^0(f)\}_\le$.

\end{definition}

\begin{table}[thp]
\caption{Examples of different signature vectors.}
\label{tab:vectorexample}
\centering
\small
\begin{tabular}{|c|c|c|}
\hline
Signatures & $f_1$ in Fig.~\ref{fig:func1} & $f_3$ in Fig.~\ref{fig:func3} \\
\hline
$OCV_1$ & (1,1,1,3,3,3) & (0,2,2,2,2,4) \\
$OCV_2$ & (0,0,0,1,1,1,1,1,1,2,2,2) & (0,0,0,0,1,1,1,1,2,2,2,2) \\
$OIV$   & (2,2,2)& (0,0,4)\\
$OSV^1$ & (0,2,2,2) & (1,1,1,1) \\
$OSV^0$ & (0,2,2,2) & (1,1,1,1) \\
$OSV$   & (0,0,2,2,2,2,2,2) & (1,1,1,1,1,1,1,1) \\
$OSDV^1$ & (0,0,0,0,0,0,0,3,0,0,0,0) & (0,0,0,4,2,0,0,0,0,0,0,0)\\
$OSDV$ & (0,0,1,0,0,0,6,6,3,0,0,0)& (0,0,0,12,12,4,0,0,0,0,0,0)\\
\hline
\end{tabular}
\end{table}

Table~\ref{tab:vectorexample} shows some examples of different signature vectors of two 3-input Boolean functions in Fig.~\ref{fig:npnequ}.
We will use this example to explain further how to get $OSDV^1$.
For $f_1$ in Fig.~\ref{fig:func1}, there is no word $X$ such that $sen^1(f, X)$=$1$ and $sen^1(f, X)$=$3$.
Moreover, only one word $X$=$111$ satisfies $sen^1(f, X)$=$0$.
Thus, $\sigma_0$=$\sigma_1$=$\sigma_3$=$(0,0,0)$.
For the three words that local sensitivity equal to 2, $011$, $101$, and $110$, we can obtain $\delta_{21}$=$0$, $\delta_{22}$=$3$ and $\delta_{23}$=$0$ according to Definition~\ref{d8}.
In summary, $OSDV^1(f_1)$=$(0,0,0,0,0,0,0,3,0,0,0,0)$.

\subsection{Signature Vectors and NPN Equivalence}
Previous work~\cite{abdollahi2008symmetry} has demonstrated that equality of $OCV_{\ell}$ is a prerequisite for NPN equivalence, so we only consider $OIV$, $OSV$, and $OSDV$ in this subsection.
The sensitive property of Boolean functions inherently considers the polarity of the output~(output negation).
For unbalanced Boolean functions, we only need to consider input phase assignment and input variables order.
That is to say, NPN equivalence is simplified to the PN equivalence problem.
Therefore, we only need to concentrate on PN equivalence to give the following theorems.

\begin{lemma}
\label{Lm1}
If Boolean function $f$ is PN-equivalent  to Boolean function $g$, that is  $f(\pi((\neg)x))=g(x)$,
then for any input $i$, we have $inf(f,\pi((\neg)i))=inf(g,i).$
\end{lemma}

\begin{proof}
% based on \ref{Lm1}. 
If  $f$ is PN-equivalent to $g$, 
it is clear that
$inf(g,i)=\frac{1}{2}|g(X) \neq g(X^{i}): X\in \{0,1\}^n|= \frac{1}{2}|f(\pi((\neg)X)\neq  f(\pi((\neg)X^{i})): X\in \{0,1\}^n)|\linebreak[1]= inf(f,\pi((\neg)i)).$
\end{proof}

%\begin{corollary}
%\label{col1}
%Two PN-equivalent functions $f$ and $g$ have the same total influence: if $f$ is PN-equivalent to $g$, then $I(f) = I(g)$.
%The contrapositive of this corollary is: if $I(f) \neq I(g)$, then $f$ is not  PN-equivalent to $g$. 
%\end{corollary}

\begin{theorem}\label{Tm1}
Two PN-equivalent functions $f$ and $g$ have the same ordered influence vector: if $f$ is PN-equivalent to $g$, then $OIV(f) = OIV(g)$.
%The contrapositive of this theorem is: if $OIV_{\ell}(f) \leq OIV_{\ell}(g), \forall \; 1 \leq \ell \leq n$, then $f \ncong g$.\added[comment={?}]{}
\end{theorem}
\begin{proof}
% based on Theorem~\ref{Tm1}.
According to Definition \ref{def:inf} and Lemma~\ref{Lm1}, it is clear that negation of a variable  will not change its influence and
$inf(f, \pi(i))\linebreak[0]=inf(g,i) $ under permutation of variable $x_i$.
Therefore, the ordered $OIV$ will not change if two functions $f$ and $g$ are PN-equivalent to each other.    
\end{proof}

%\subsection{Sensitivity Signatures and NPN Equivalence}

\begin{lemma}
\label{Lm2} 
If Boolean function $f$ is  PN-equivalent  to Boolean function $g$, that is  $f(\pi((\neg)x)=g(x)$,
then for any input $X$,   we have $sen(f,\pi((\neg)X))=sen(g,X).$ 
\end{lemma}

\begin{proof}
Since $f(\pi((\neg)x_1,(\neg)x_2,\cdots,(\neg)x_n))=g(x_1,x_2,\\ \cdots, x_n)$, it is clear that 
if $f$ is sensitive at index $i$ for input word $\pi((\neg)X)$,
then $g$ is sensitive at index $j$ for input word $X$ where $\pi(j)$=$i$.
The sensitive property inherently considers negations of the variables so that flipping an input can not change anything of a Boolean function's sensitivity.
%\added[comment={取反不改变影响力和敏感度，可以独立成一条引理}]{}

For example, let  $f(x)$ be a 4-bit Boolean function, permutation $\pi(x_1x_2x_3x_4)$ = $x_4x_3x_2x_1$, selective negation $(\neg)x_1x_2x_3x_4$ = $\overline{x_1}x_2\overline{x_3}x_4$, and $f(\pi((\neg)x_1x_2x_3x_4))$ = $f(x_4,\overline{x_3},x_2,\overline{x_1})$ = $g(x_1,x_2,x_3,x_4)$.
%}{$\pi(1,2,3,4)$=$(4,3,2,1)$ and $f(\pi(\overline{x_1}x_2 \overline{x_3}x_4)$=\\$g(x_1x_2x_3x_4)$.}
Assume that $f$ is sensitive at index $2$ for word $(x_4,\overline{x_3},x_2,\overline{x_1})$  (i.e., at where $\overline{x_3}$ locates), we have $g(x_1,x_2,x_3,x_4)$ = $f(x_4,\overline{x_3},x_2,\overline{x_1})$ = $\neg f(x_4,x_3,x_2,\overline{x_1})$ = $\neg g(x_1,x_2,\overline{x_3},x_4)$.
%}{and input $\pi(\overline{x_1}x_2\overline{x_3}x_4)$=$x_4\overline{x_3}\\x_2\overline{x_1}$ is sensitive on index $2$, we have $f(\pi(\overline{x_1}  x_2\overline{x_3}x_4))$=$f\\(x_4  \overline{x_3}x_2\overline{x_1})$=$g(x_1x_2x_3x_4)$ and $ \neg g(x_1 x_2x_3x_4)$=$ \neg f(x_4 \overline{x_3}\\x_2\overline{x_1})$=$f(x_4x_3x_2\overline{x_1})$=$g(x_1x_2\overline{x_3}x_4)$.}
Function $g$ is sensitive at index $3$=$\pi(2)$ for input word $(x_1,x_2,x_3,x_4)$.

Therefore, for any $X$, $sen(f,\pi((\neg)X))=sen(g,X).$  
\end{proof}

\begin{figure}[tbp]
\centering
\begin{subfigure}[b]{0.49\columnwidth}
\centering
\includegraphics[scale=0.25]{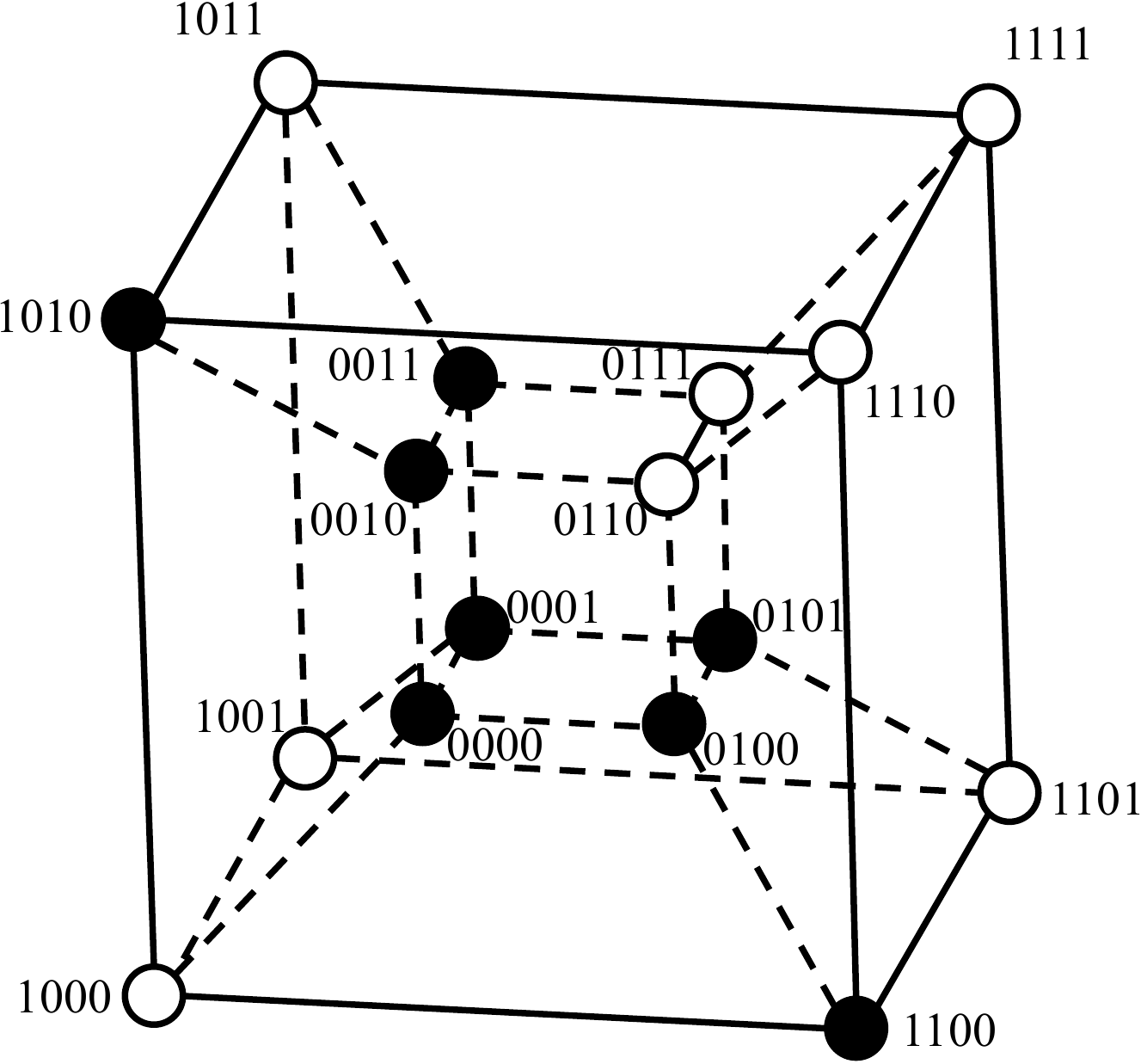}
\caption{\tiny{$OSV^1(f)=\{1,1,1,1,2,2,3,3\}$\\$OSV^0(f)=\{0,1,2,2,2,2,2,3\}$.}}
\label{fig:4balanced1}
\end{subfigure}
\begin{subfigure}[b]{0.49\columnwidth}
\centering
\includegraphics[scale=0.25]{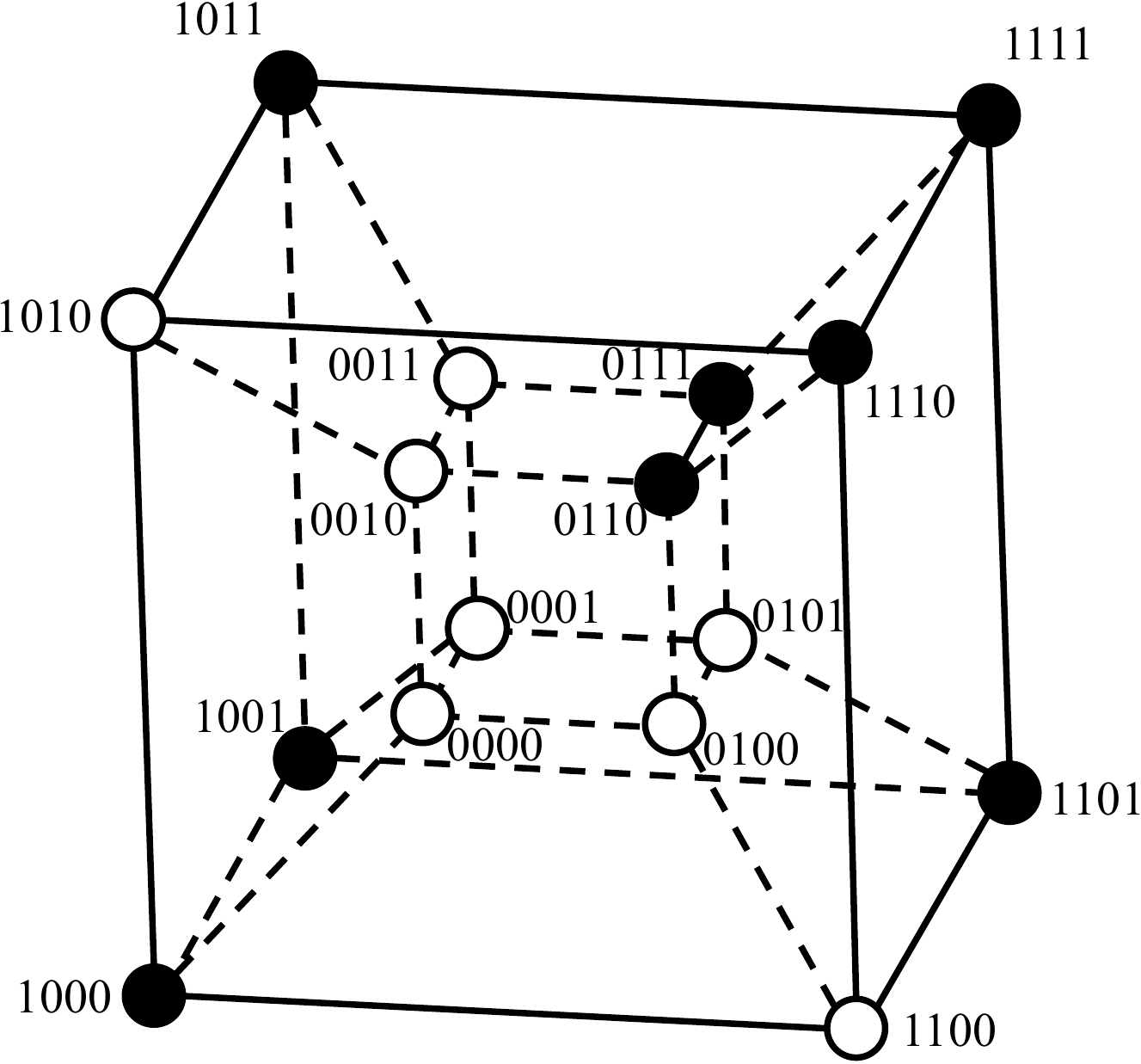}
\caption{\tiny{$OSV^1(g)=\{0,1,2,2,2,2,2,3\}$\\$OSV^0(g)=\{1,1,1,1,2,2,3,3\}$.}}
\label{fig:4balanced2}
\end{subfigure}
\caption{Two NPN equivalent balanced Boolean functions. $OSV^1(f)=OSV^0(g)$ and $OSV^0(f)=OSV^1(g)$ in these two functions.}
\label{fig:4balanced}
\vspace{-1em}
\end{figure}

\begin{theorem}\label{Tm2}
Two PN-equivalent unbalanced functions $f$ and $g$ have the same ordered sensitivity vector, ordered 0-sensitivity vector, and ordered 1-sensitivity vector: if $f$ is PN-equivalent to $g$, then $(OSV, OSV^0, OSV^1)(f) = (OSV, OSV^0, OSV^1)(g)$.
%The contrapositive of this theorem is: if $OSV(f) \neq OSV(g), OSV^0(f) \neq OSV^0(g)$, or $OSV^1(f) \neq OSV^1(g)$, then $f \ncong g$.
\end{theorem}
A similar theorem has been proved by Zhang~{\it et al.}~\cite{zhang2021enhanced}.
However, they ignored balanced Boolean functions.
The two Boolean functions in Fig.~\ref{fig:4balanced} are NPN equivalent.
For these two functions, $OSV^1(f)$=$OSV^0(g)$ and $OSV^0(f)$=$OSV^1(g)$.
We split 1-sensitivity and 0-sensitivity to handle balanced Boolean functions. Given two NPN-equivalent unbalanced Boolean functions $f$ and $g$, it is easy to check whether $g$ is transformed from $f$ by negation. 
We can use the 0-ary cofactor of the functions to find out the potential negation. 
However, if $f$ and $g$ are balanced, it cannot find out the potential negation just by using their 0-ary cofactor.  
In such cases, we calculate both 1-sensitivity and 0-sensitivity of the functions to deal with the potential negation.   

\begin{theorem}
\label{th4.3}
Given two balanced Boolean functions $f$ and $g$, if $f$ is NPN-equivalent to  $g$, then $OSV^1(f)$=$OSV^1(g)$, $OSV^0(f)$=$OSV^0(g)$ or $OSV^1(f)$=$OSV^0(g)$, $OSV^0(f)\\$=$OSV^1(g)$.
\end{theorem}

\begin{proof}
According to Theorem \ref{Tm2}, if $f$ PN-equivalent to $g$, that is $f(\pi((\neg)x))=g(x)$, we have $OSV^1(f)=OSV^1(g)$ and  $OSV^0(f)=OSV^0(g)$.  
If $\neg f(\pi((\neg)x))=g(x)$, it is clear that  $OSV^0(f)=OSV^1(g)$ and  $OSV^1(f)=OSV^0(g)$.
\end{proof}
In order to deal with balanced Boolean functions in our algorithms,   if $OSV^1(f)$ is smaller than $OSV^0(f)$, we will exchange these two vectors and always put the smaller one in $OSV^0(f)$.

\begin{lemma}
\label{Lm3}
Given two inputs $X$ and $Y$, if $f$ is PN-equivalent to $g$, then $\langle sen(f,\pi((\neg)X))$, $sen(f,\pi((\neg)Y))$,   $h(\pi((\neg)X)$,  $\pi((\neg)Y))  \rangle$=$\langle sen(g,X)$, $sen(g,Y)$, $h(X,Y)  \rangle$.
\end{lemma}

\begin{proof}
It is easy to see that $h(X,Y)=h(\pi((\neg)X,(\neg)Y))$. 
According to Lemma \ref{Lm2} and Definition~\ref{d7}, the theorem holds.
\end{proof}

\begin{theorem}
\label{Tm3}
Two PN-equivalent unbalanced functions $f$ and $g$ have the same ordered sensitivity distance vector, ordered 0-sensitivity distance vector, and ordered 1-sensitivity distance vector: if $f$ is PN-equivalent to $g$, then $(OSDV, OSDV^0, OSDV^1)(f) = (OSDV, OSDV^0, OSDV^1)(g)$.
For two balanced Boolean functions $f$ and $g$, if $f$ is NPN-equivalent to  $g$, then $OSDV^1(f)=OSDV^1(g)$, $OSDV^0(f)=OSDV^0(g)$ or $OSDV^1(f)=OSDV^0(g)$, $OSDV^0(f)=OSDV^1(g)$.
\end{theorem}

\begin{proof}

Given two inputs $x$ and $y$, if $f$ is PN-equivalent to $g$, according to  Lemma \ref{Lm3},   we have that $\langle h(X,Y)$, $sen(f,X)$, $sen(f,Y)\rangle$=$\langle h(X',Y')$, $sen(g,X')$, $sen(g,Y')\rangle$, where $X'$ and $Y'$ is the NP transformation of $X$ and $Y$. 
According to Definition \ref{d7} and  Definition \ref{d8}, it is clear that for unbalanced Boolean functions $f$ and $g$, we have $(OSDV,OSDV^0,OSDV^1)(f)$ = $(OSDV,OSDV^0,OSDV^1)(g)$. 
Similar to the proof of Theorem \ref{th4.3}, we can prove the results of the rest of the theorem for balanced Boolean functions.  
\end{proof}

\begin{figure*}[htbp]
\centering
\begin{subfigure}[b]{0.49\columnwidth}
\centering
\includegraphics[scale=0.25]{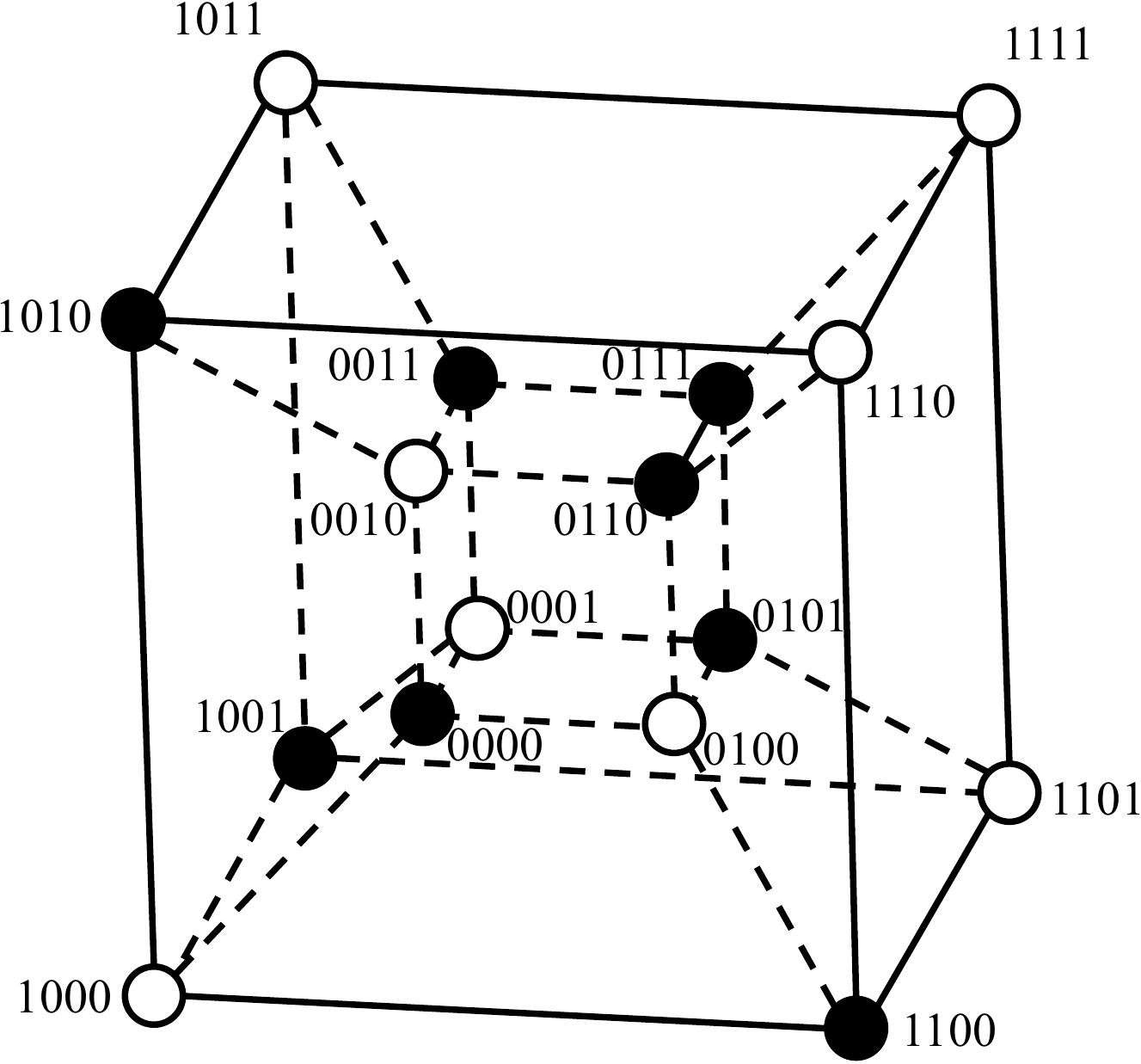}
\caption{$g_1$.}
\label{fig:4dimf}
\end{subfigure}
\begin{subfigure}[b]{0.49\columnwidth}
\centering
\includegraphics[scale=0.25]{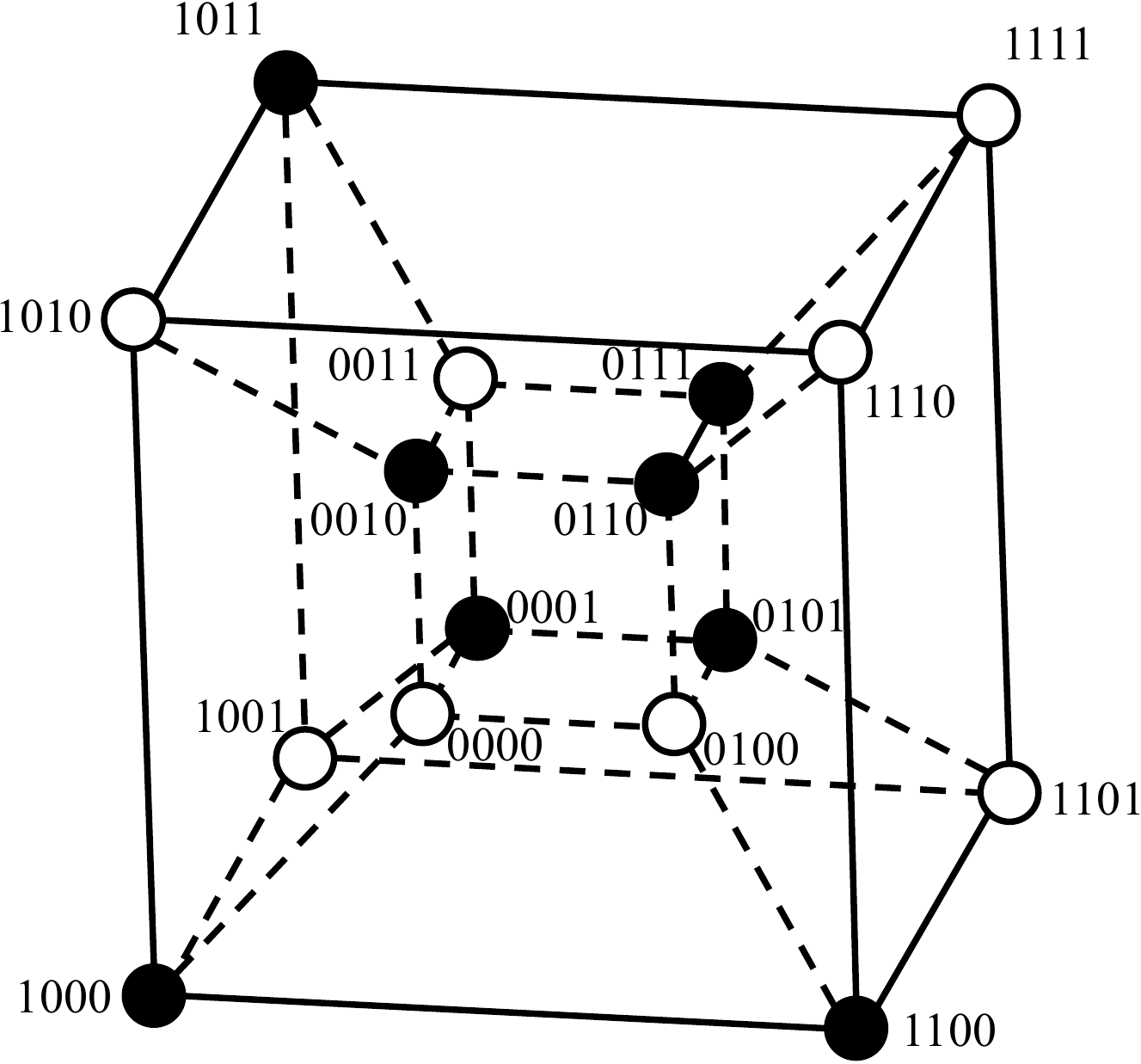}
\caption{$g_2$.}
\label{fig:4dimg}
\end{subfigure}
\begin{subfigure}[b]{0.49\columnwidth}
\centering
\includegraphics[scale=0.25]{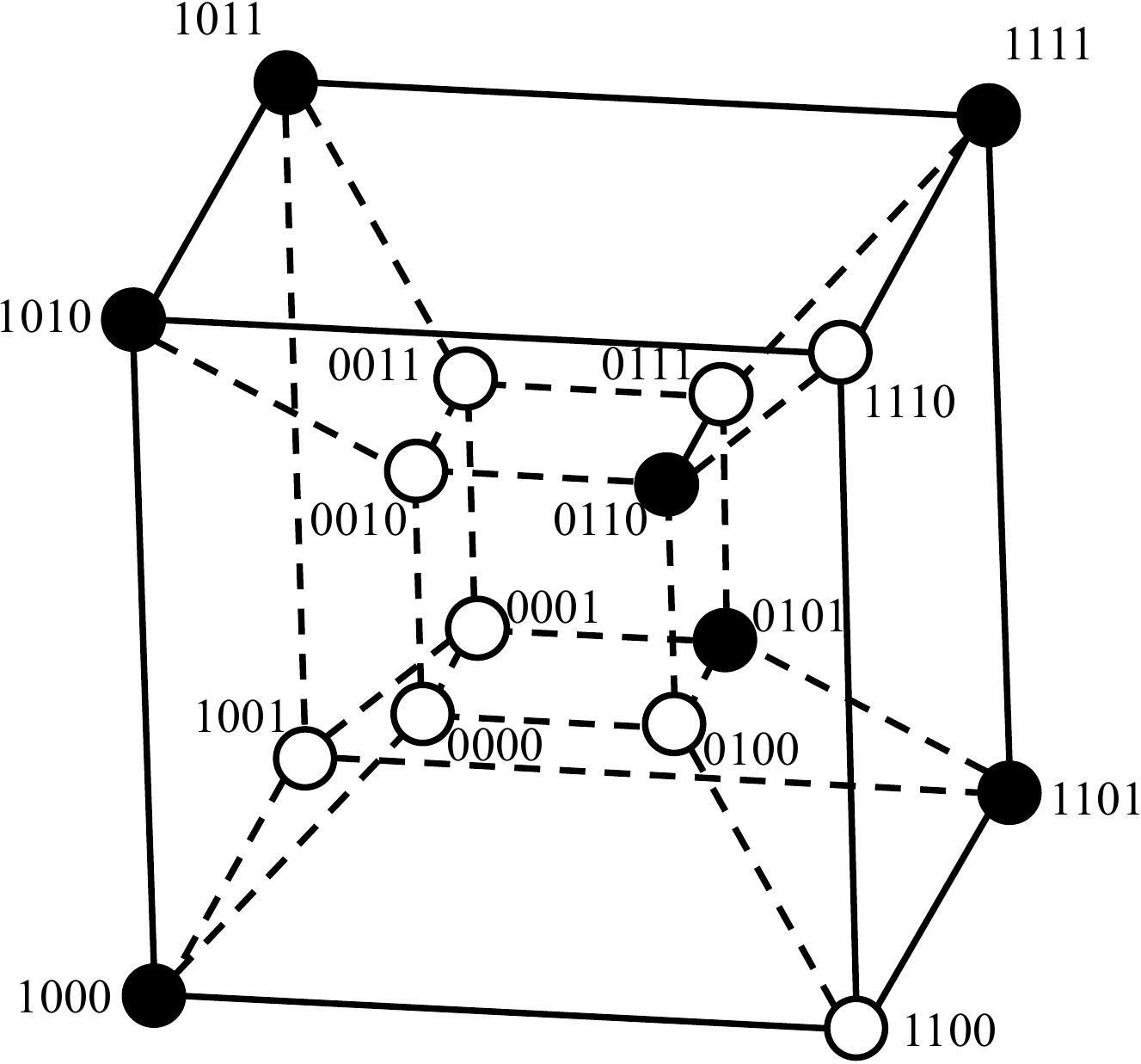}
\caption{$h_1$.}
\label{fig:4dimh}
\end{subfigure}
\begin{subfigure}[b]{0.49\columnwidth}
\centering
\includegraphics[scale=0.25]{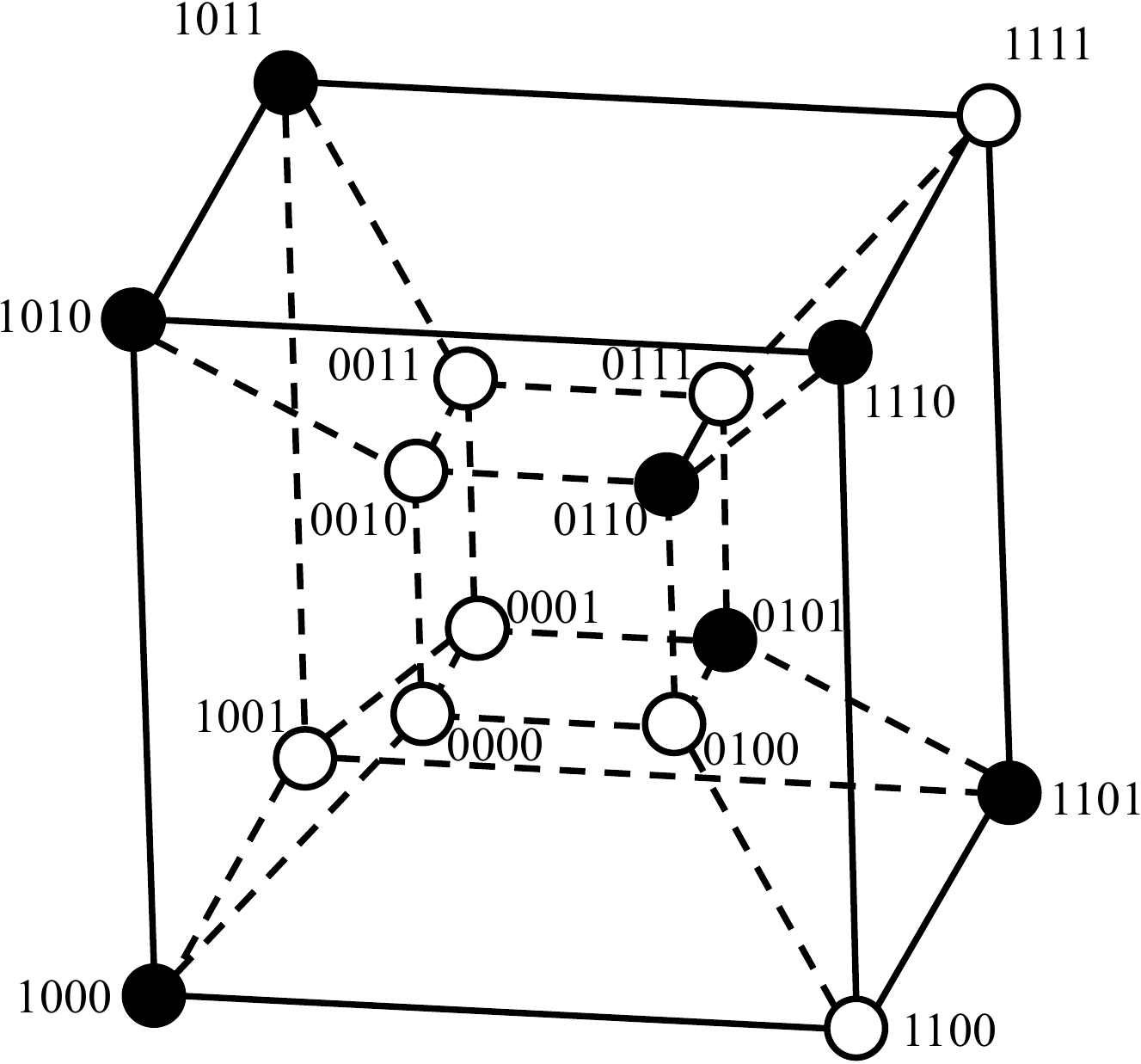}
\caption{$h_2$.}
\label{fig:4dimi}
\end{subfigure}
\caption{Hypercubes of two pairs of nonequivalent 4-input Boolean functions $g_1$, $g_2$ and $h_1$, $h_2$. }
\label{fig:4dim}
\end{figure*}

\section{Classifier}
\label{sec:canonical}
In this section,  we first show the effect of influence and sensitivity signature vectors on NPN classification and then present our classification algorithm.

\subsection{Signature Vectors Selection}
\label{sec:sigbasedcan}

All aries of cofactor signature vectors have been proved to be a canonical form~\cite{abdollahi2008symmetry}.
However, computing all-ary cofactor signatures are time-consuming.
Next, we will show that $OIV$, $OSV$, and $OSDV$ are strong discriminators of NPN non-equivalence over $OCV_1$ and $OCV_2$.

Fig.~\ref{fig:4dim} shows four hypercubes of two pairs of nonequivalent 4-input Boolean functions $g_1$, $g_2$ and $h_1$, $h_2$.
The $\mathit{OC}V_1(g_1)$=$\mathit{OC}V_1(g_2)$=$(3,4,4,4,4,4,4,5)$ and  $\mathit{OC}V_2(g_1)$=$\mathit{OC}V_2(g_2)$=$(1,1,1,2,2,2,2,2,2,2,2,2,2,2,2,\\2,2,2,2,2,2,3,3,3)$ are identical, respectively.
However, $OIV_1(g_1)$=$(6,6,6,8)$ and $OIV_2(g_2)$=$(2,6,6,8)$ of these two functions are different.
The ordered influence vector can distinguish nonequivalent Boolean functions that cannot be classified by $OCV_1$ and $OCV_2$.
Moreover, the $\mathit{OC}V_1(h_1)$=$\mathit{OC}V_1(h_2)$=$(2,3,3,3,4,4,4,5)$,   $\mathit{OC}V_2(h_1)$=$\mathit{OC}V_2(h_2)$=$(0,1,1,1,1,1,1,1,1,2,2,2,2,2,2,\\2,2,2,2,2,3,3,3,3)$ and $\mathit{OI}V_1(h_1)$=$\mathit{OI}V_1(h_2)$=$(3,5,5,5)$ are identical, respectively.
But $OSV^1(h_1)$=$(2,2,2,2,3,3,\\4)$ and $OSV^1(h_2)$=$(1,2,3,3,3,3,3)$ of these two functions are different.
Therefore, the ordered sensitivity vector can distinguish nonequivalent Boolean functions that can not be classified by $OCV_1$ and $OCV_2$.
More evaluations will be seen in Section~\ref{sec:evalsemi}.

\subsection{Classification Method}
\label{sec:computation}
As the property of cofactor, influence, and sensitivity mentioned above, it is very convenient to implement the computation of the signature based on the binary string as the representation. 
We adopt several bitwise operation techniques in~\cite{warren2013hacker} for fast signatures computation.
%A mask defines which bits to keep or clear after bitwise operations for signature computation acceleration.
For example, to calculate the cofactor signature of a certain literal, we only need to keep the relevant bits in the truth table and count the number of 1s.
%Masks for $\ell$-ary cofactor, influence, and sensitivity computation are denoted as $MC_\ell$, $MI$ and $MS$.

\setlength{\textfloatsep}{1pt}
\begin{algorithm}[tbp]
\caption{NPN Classifer}
\label{algorithm:npnclass}
\begin{algorithmic}[1]
    % \Require input variables $n$, truth table set $tts$, masks: $MI_1$, $MC_1$, $MC_2$, $MS$
    \Require input variables $n$, truth table set $tts$
    \Ensure NPN equivalent classes
    \For{$tt$ in $tts$}
        \State get $OCV_1(tt)$ and $OCV_2(tt)$; 
        \State get $OIV(tt)$;
        \State get $OSV(tt)$;
        \State compute $OSDV(tt)$ using $OSV(tt)$;
        \State construct $MSV(tt)$ using $OCV_1(tt)$, $OCV_2(tt)$, $OIV_1(tt)$, $OSV(tt)$ and $OSDV(tt)$;
        \State $class \gets hash(MSV(tt))$;
    \EndFor
\end{algorithmic}
\end{algorithm}

Algorithm~\ref{algorithm:npnclass} gives the overall algorithm for NPN classification. 
First, we compute several signature vectors as the definitions above~(line 2 to line 5).
%$MI$, $MC$ focus on the dimension of the number of inputs. 
%$MI$ only considers every dimension once, $MC_1$ also considers every dimension's complementary. 
%$MC_2$ considers two dimensions of every dimension's combination. 
%$MS$ focuses on the position at the truth table. 
Then, we construct the MSV (Mixed Signature Vector)~(line 6).
At last, a hash function is used to finish the classification~(line 7) of this Boolean function.
For runtime saving, we can use $OSV^1$, $OSV^0$ replacing $OSV$ and $OSDV^1$, $OSDV^0$ replacing $OSDV$.
When considering balanced Boolean functions, it should be noted that different $OSV$ need to be constructed according to the Theorem~\ref{Tm3}.
%Theorem~\ref{th4.3} and Theomre~\ref{Tm3} to guarantee the canonical property of $MV$.
%Considering the trade-off between accuracy and runtime, we use $MI, MC_1, MC_2, MS$ here. Firstly, we speed up the computation of signatures by pre-computed masks. 
Most NPN classification methods define a canonical form utilizing cofactor signatures and hierarchical symmetry properties first and then propose a complex algorithm to compute it.
Boolean functions with different symmetric properties further increase the complexity.
Our classification method only needs bitwise operations and hash to finish the classification.

\section{Experimental Evaluations}
\label{sec:evaluation}

\begin{table*}
\caption{The results of classification using different signature vectors.}
\label{tab:canonform}
\centering
\small
\begin{tabular}{IcIcIcIcIcIcIcIcIcIcI}
\shline
n   & \makecell[c]{\#Exact \\Classes} & 
\makecell[c]{\#Classes \\ by $OIV$}  & 
\makecell[c]{\#Classes \\ by $OCV_1$}  & 
\makecell[c]{\#Classes \\ by $OSV$} & 
\makecell[c]{\#Classes by \\ $OIV$+$OSV$} & 
\makecell[c]{\#Classes by \\ $OCV_1$+$OSV$} & 
\makecell[c]{\#Classes by \\$OCV_1$+$OCV_2$\\+$OSV$} & 
\makecell[c]{\#Classes by \\ $OIV$+$OSV$\\+$OSDV$} & 
\makecell[c]{\#Classes \\by All} \\
\shline
4   & 49      & 28      & 41      & 48      & 48        & 49         & 49       & 49         &     \textbf{49}     \\
5   & 312     & 173     & 251     & 305     & 310       & 311        & 311      & 312        &     \textbf{312}    \\
6   & 1673    & 1175    & 1380    & 1619    & 1654      & 1668       & 1671     & 1673       &     \textbf{1673}   \\
7   & 6071    & 5224    & 5498    & 5985    & 6052      & 6057       & 6057     & 6052       &     \textbf{6071}   \\
8   & 48895   & 44497   & 44183   & 48584   & 48876     & 48876      & 48876    & 48877      &     48887  \\
9   & 92741   & 87485   & 87080   & 92381   & 92721     & 92723      & 92723    & 92721      &     92725  \\
10  & 184832  & 178155  & 177799  & 184428  & 184794    & 184794     & 184795   & 184796     &     184796 \\

\shline
\end{tabular}
\end{table*}

\subsection{Setup}
We implement our classifier algorithm in C++ and compare our work with some state-of-the-art NPN classification works.
All procedure runs on an Intel Xeon 2-CPU 20-core computer with 60GB RAM. 
The input of a benchmark consists of a list of Boolean functions (in the truth-table form) to be classified under NPN equivalence. 
And the output consists of the number of equivalence classes, as well as the classified truth tables.

We use EPFL benchmarks~\cite{amaru2015epfl} to test the effectiveness of our algorithm on real synthesis applications.
The truth tables are extracted from these benchmarks using cut enumeration.
We deleted the Boolean functions of the same truth table.

\subsection{Evaluation of the Signature Vectors}
\label{sec:evalsemi}
We evaluate the effectiveness of each signature vector part and different combinations.
Table~\ref{tab:canonform} shows the results.
The number of exact classes in the third column is run using Kitty when $n \leq 6$ and the exact version in~\cite{abc} when $n > 6$.
In general, cofactor signatures are more effective in classification than influence but worse than sensitivity.
The combination of influence and sensitivity signature vectors is better than cofactor signature vectors.
It performs an exact classification when $n \leq 7$.
Therefore, point characteristics and face characteristics have different properties, and their combination can effectively complete NPN classification. 

\begin{table*}[tbhp]
\caption{Runtime and accuracy comparison of different NPN classifiers.}
\label{tab:npnmatch}
\centering
\small
\begin{tabular}{IcIcIcIcIcIcIcIcIcIcIcIcIcI}
\shline
\multirow{2}{*}{n}   & \multirow{2}{*}{\#Func} & \multirow{2}{*}{\makecell[c]{\#Exact \\Class}}  & \multicolumn{2}{cI}{Kitty} &\multicolumn{2}{cI}{testnpn -6~\cite{huang2013fast}} & \multicolumn{2}{cI}{testnpn -7~\cite{petkovska2016fast}} &\multicolumn{2}{cI}{testnpn -11~\cite{zhou2020fast}} & \multicolumn{2}{cI}{\textbf{Ours}} \\
 \cline{4-13}
& &  &	\#Class	&  Time  &	\#Class  &	Time	 & \#Class & Time  &	\#Class	&  Time &	\#Class	&  Time \\
\shline
4 & 1146    & 49     & 49     & 0.031  & 251    & 0.0003  & -      & -         & 52        & 0.0024 & \textbf{49}     & \textbf{0.0013} \\
5 & 6824    & 312    & 312    & 0.858  & 1586   & 0.0026  & -      & -         & 322       & 0.015 & \textbf{312}    & \textbf{0.0049}\\
6 & 28672   & 1673   & 1673   & 39.453   & 7375   & 0.006  & 1752   & 0.021  & 1690      & 0.046 & \textbf{1673}   & 0.121 \\
7 & 80123   & 6071   & -      & -	        & 23318  & 0.216  & 6249   & 0.067  & 6115      & 0.194 & \textbf{6071}   & 0.773 \\
8 & 480516  & 48895  & -      & -	        & 190708 & 0.13  & 50066  & 0.554  & 49577     & 4.701    & 48887  & 12.350\\
9 & 691474  & 92741  & -      & -         & 278090 & 0.296  & 94283  & 1.438   & 93575     & 128.926  & 92725  & \textbf{51.86}\\
10 & 1153464 & 184832 & -      & -         & 500911 & 0.881  & 187117 & 4.193  & 186098    & 1329.995 & 184796 & \textbf{318.56}\\
%&11 & 1153464 & 184832 & -      & -         & 500911 & 0.880915  & 187117 & 4.192923  & 186098    & 16.9583  & 184796 & 461.151\\
%&12 & 1153464 & 184832 & -      & -         & 500911 & 0.880915  & 187117 & 4.192923  & 186098    & 16.9583  & 184796 & 461.151\\
%&13 & 1153464 & 184832 & -      & -         & 500911 & 0.880915  & 187117 & 4.192923  & 186098    & 16.9583  & 184796 & 461.151\\
%&14 & 1153464 & 184832 & -      & -         & 500911 & 0.880915  & 187117 & 4.192923  & 186098    & 16.9583  & 184796 & 461.151\\
%&15 & 1153464 & 184832 & -      & -         & 500911 & 0.880915  & 187117 & 4.192923  & 186098    & 16.9583  & 184796 & 461.151\\
%&16 & 1153464 & 184832 & -      & -         & 500911 & 0.880915  & 187117 & 4.192923  & 186098    & 16.9583  & 184796 & 461.151\\

\shline
\end{tabular}
\end{table*}

\subsection{Evaluation of NPN Classifier}
\label{sec:evalmatch}
We evaluate our classification algorithm and compare the results with Kitty in EPFL logic synthesis libraries~\cite{soeken2018epfl} and several state-of-the-art works~\cite{huang2013fast,petkovska2016fast,zhou2020fast}, which are implemented in ABC~\cite{abc} as command \cmd{testnpn} with different arguments.
Due to some methods in~\cite{zhou2020fast} using an exhaustive enumeration for exact classification at the end, we modified ABC and removed this part for a fair comparison.
Table~\ref{tab:npnmatch} shows the results.
Kitty can gain exact NPN classification, but it runs slowly and will not work when $n>6$.
The command \cmd{testnpn -7} fails when $n \leq 5$. 
So we omit these parts of the results.
The number of exact classes runs using Kitty when $n \leq 6$ and the exact version in~\cite{abc} when $n > 6$.
Our classification gains up to 325x speedup over Kitty~(6-bit) with the same classification accuracy.
Although the algorithm in~\cite{huang2013fast} is ultra-fast, it fails inaccurate classification.
Algorithms in~\cite{petkovska2016fast} and~\cite{zhou2020fast} show speed and classification accuracy improvements.

Previous classification methods need a canonical form and computation algorithm to get NPN classes.
The runtime of such methods is related to the properties of Boolean functions, such as symmetries.
If the Boolean functions have bad properties, then the computation of the canonical form will become complicated, leading to unstable runtime.
The classifier proposed in this paper needs only bitwise computation and hash operations, so the runtime is only related to the bitwidth and the total number of Boolean functions.
Fig.~\ref{fig:stable} shows the runtime of our classifier and \cmd{testnpn -11} testing on randomly generated 5-bit and 7-bit Boolean functions.
The horizontal axis is the number of generated Boolean functions.
We randomly generate a fixed number of Boolean functions with truth tables in consecutive binary encoding for each bit.
This figure shows that the runtime of our classifier is almost linear with the total number of Boolean functions, while \cmd{testnpn -11} fluctuates widely with different sets of Boolean functions.
Actually, our classifier cannot return exact matching solutions.
Influence and sensitivity still have great potential to be extended to the traditional method to achieve exact NPN classification, and we will explore them in the future.

\begin{figure}[tbp]
\centering
\includegraphics[scale=0.19]{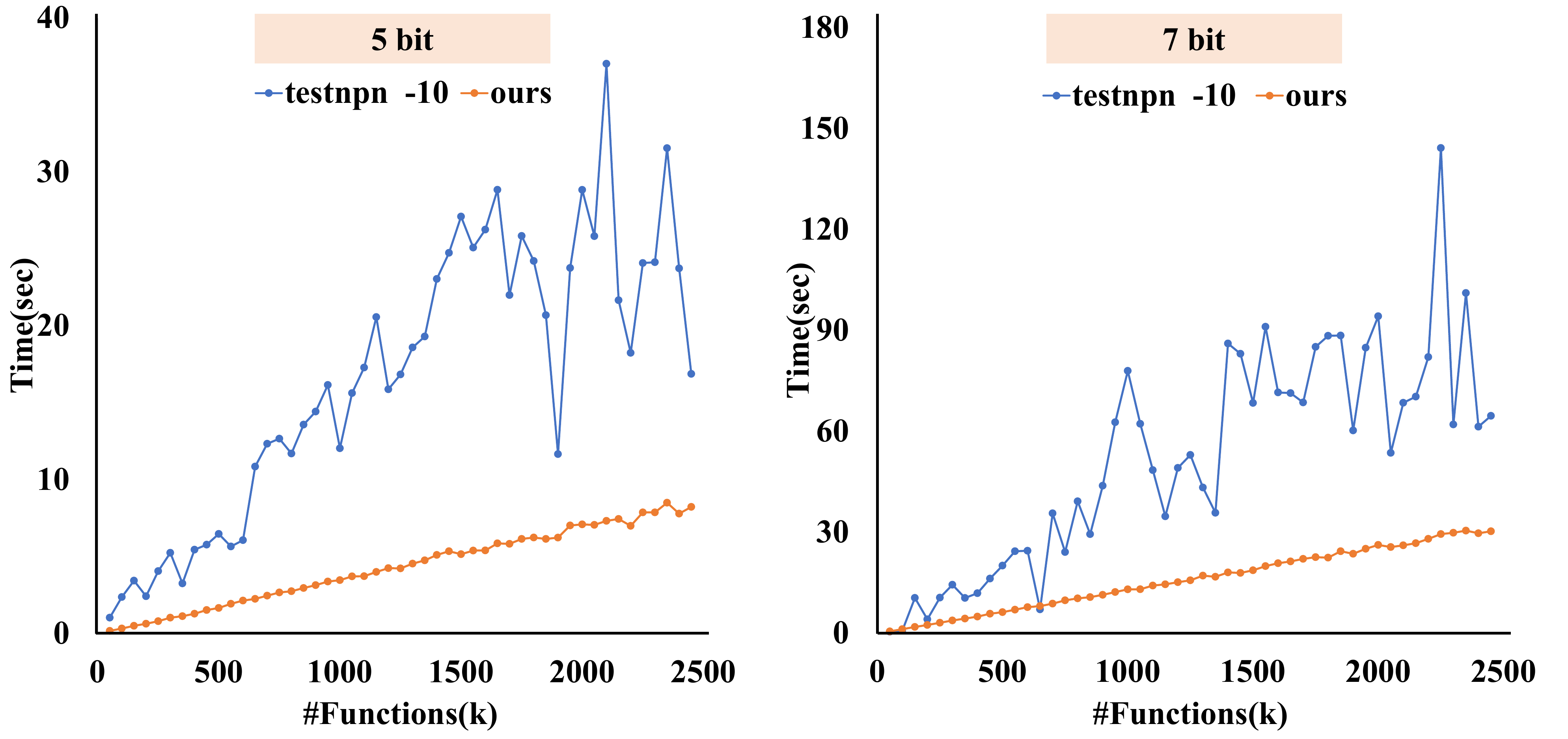}
\label{fig:6bit}
\caption{Our classifier has stable runtime.}
\label{fig:stable}
\end{figure}

\section{Conclusion}
\label{sec:conclusion}

This paper rethought the NPN classification problem in terms of face and point characteristics of Boolean functions.
We introduced Boolean sensitivity and influence, two concepts considering point and face-point characteristics in NPN classification.
We designed some signature vectors based on these two characteristics.
Combined with cofactor signatures, we develop a new classifier that only relies on signature vector computation.
The experiments showed that the proposed NPN classifier gains better NPN classification accuracy with comparable and stable speed.

\section*{Acknowledgment}
This work is partly supported by the National Natural Science Foundation of China (Grant No. 62090021) and the National Key R\&D Program of China (Grant No. 2022YFB4500500).
Shenggen Zheng acknowledges support in part from the Major Key Project of PCL.

%\balance

\bibliographystyle{IEEEtran}
\bibliography{short}

%\unbalance

\end{CJK*}
\end{document}